\documentclass[12pt]{article}
\usepackage{amsmath,amsthm,amssymb,fullpage,url,cite,cleveref,bm,enumerate}
\newcommand{\remove}[1]{}
\newtheorem{theorem}{Theorem}[section]

\newtheorem{lemma}[theorem]{Lemma}

\newtheorem{definition}[theorem]{Definition}
\newtheorem{corollary}[theorem]{Corollary}
\newtheorem{conjecture}[theorem]{Conjecture}

\newcommand{\F}{\mathbb{F}}

\newcommand{\N}{\mathbb{N}}
\newcommand{\A}{\bold{a}}

\newcommand{\cS}{\mathcal{S}}
\newcommand{\cV}{\mathcal{V}}
\newcommand{\etal} {{et al. }}

\title{MDS matrices over small fields:\\
A proof of the GM-MDS conjecture}
\author{Shachar Lovett\thanks{Department of Computer Science and Engineering, University of California, San Diego. 
{\tt slovett@cs.ucsd.edu.} Research supported by NSF CAREER award 1350481 and CCF award 1614023.}}

\begin{document}
\maketitle

\begin{abstract}
An MDS matrix is a matrix whose minors all have full rank. A question arising in coding theory is what
zero patterns can MDS matrices have. There is a natural combinatorial characterization (called the MDS condition)
which is necessary over any field, as well as sufficient over very large fields by a probabilistic argument.

Dau et al. (ISIT 2014) conjectured that the MDS condition is sufficient over small fields as well, where
the construction of the matrix is algebraic instead of probabilistic. This is known as the GM-MDS conjecture. Concretely,
if a $k \times n$ zero pattern satisfies the MDS condition, then they conjecture that there exists an MDS matrix with this zero pattern
over any field of size $|\F| \ge n+k-1$. In recent years, this conjecture was proven in several special cases.
In this work, we resolve the conjecture.
\end{abstract}

\section{Introduction}

An MDS matrix is a matrix whose minors all have full rank.
These matrices arise naturally in coding theory, as they are generating matrices for MDS (Maximally Distance Separable) codes.
A question arising in coding theory, motivated by applications in multiple access networks \cite{halbawi2014distributed,dau2015simple}
and in secure data exchange \cite{yan2013algorithms,yan2014weakly},
is what zero patterns can MDS matrices have. Namely, how sparse can MDS matrices be?

There is a natural combinatorial characterization on the allowed zero patterns, called the MDS condition.
Let $A$ be a $k \times n$ MDS matrix with $k \le n$.
We can describe its zero/nonzero pattern by a set system $S_1,\ldots,S_k \subset [n]$, where $S_i=\{j \in [n]: A_{i,j}=0\}$.

There are several restrictions on the structure of such set systems. Clearly, any row of $A$ can have at most $k-1$ zeros,
so $|S_i| \le k-1$ for all $i$. Similarly, any two rows of $A$ can have at most $k-2$ common zeros, so $|S_i \cap S_j| \le k-2$
for all $i \ne j$. In general, this is known as the \emph{MDS condition} on the set system:
\begin{equation}
\tag{$\star$}
|I|+\left| \bigcap_{i \in I} S_i \right| \le k \qquad \forall I \subseteq [k], I \text{ nonempty}.
\end{equation}

It is known that the MDS condition is also sufficient for the existence of MDS matrices with zero pattern given by the
set system, if the underlying field is large enough.
Concretely, let $S_1,\ldots,S_k \subset [n]$ be a set system which satisfies the MDS condition.
Let $\F$ be the underlying field, and assume that $|\F| > {n \choose k}$. Let $A$ be a randomly chosen $k \times n$
matrix over $\F$, where $A_{i,j}=0$ if $j \in S_i$, and otherwise $A_{i,j} \in \F$ is chosen uniformly and independently.
Such a matrix $A$ is an MDS matrix with positive probability. The reason is that the number
of maximal $k \times k$ minors of $A$ is ${n \choose k}$, and the MDS condition implies that the determinants of these minors are not identically zero. So, each minor has a probability of $|\F|^{-1}$ to be singular, and by the union bound, if $|\F| > {n \choose k}$, then with positive
probability all minors are nonsingular. This bound was improved to $|\F| > {n-1 \choose k-1}$ in \cite{dau2013balanced}.

Dau {et al.} \cite{dau2014existence} conjectured that the MDS condition is sufficient over small fields as well.
This is known as the GM-MDS conjecture. Concretely,
if a $k \times n$ zero pattern satisfies the MDS condition, then there exists an MDS matrix with this zero pattern
over any field of size $|\F| \ge n+k-1$. Clearly, if this is true then a different argument than the probabilistic
argument above would be needed.

\begin{conjecture}[GM-MDS conjecture \cite{dau2014existence}]
\label{conj:gmmds}
Let $S_1,\ldots,S_k \subset [n]$ be a set system which satisfies the MDS condition.
Then for any field $\F$ with $|\F| \ge n+k-1$, there exists a $k \times n$ MDS matrix $A$ over $\F$ with $A_{i,j}=0$ whenever $j \in S_i$.
\end{conjecture}

We prove \Cref{conj:gmmds} in this work.

\begin{theorem}
\label{thm:main}
\Cref{conj:gmmds} is correct.
\end{theorem}

First, we describe an algebraic framework introduced by Dau \etal \cite{dau2014existence}
towards proving \Cref{conj:gmmds}.

\subsection{The algebraic GM-MDS conjecture}

Dau \etal \cite{dau2014existence} formulated an algebraic conjecture that implies \Cref{conj:gmmds}:
if $S_1,\ldots,S_k$ is a set system that satisfies the MDS condition, then there exists a Generalized Reed-Muller code
with zeros in locations prescribed by the set system. Otherwise put, the matrix $A$ can be factored as the product of a $k \times k$ invertible matrix
and a $k \times n$ Vandermonde matrix. Before explaining these ideas further, we first set up some notations.

Let $\F$ be a finite field, and let $x,a_1,\ldots,a_n$ be formal variables,
where we shorthand $\A=(a_1,\ldots,a_n)$. We use the standard notations $\F[\A,x]$ for the ring of polynomials over $\F$ in the variables
$\A,x$; $\F(\A)$ for the field of rational functions over $\F[\A]$; and $\F(\A)[x]$ for the ring of univariate polynomials in $x$ over $\F(\A)$.
Given a set $S \subset [n]$ define a polynomial $p=p(S) \in \F[\A,x]$ as follows:
$$
p(\A,x):=\prod_{i \in S} (x-a_i).
$$
Given a set system $\cS=\{S_1,\ldots,S_k\}$ define $P(\cS):=\{p(S_1),\ldots,p(S_k)\}$.

Let $\cS=\{S_1,\ldots,S_k\}$ be a set system which satisfies the MDS condition. It is possible to assume without loss of generality
that each $S_i$ is maximal, namely that $|S_i|=k-1$ for all $i \in [k]$.
For example, if we are allowed to increase $n$ then we can replace each $S_i$ with $S_i \cup T_i$
where $|T_i|=k-1-|S_i|$ and $T_1,\ldots,T_k,[n]$ are pairwise disjoint. An improved reduction is given in \cite{dau2014existence}
which does not require increasing $n$.

Either way, under this assumption the polynomials $P(\cS)$ form a set of $k$ polynomials of degree $k-1$,
which we denote by $p_1,\ldots,p_k$.
Define the $k \times n$ matrix $A$ as
$A_{i,j} = p_i(a_j)$. Note that entries of $A$ are polynomials in $\F[\A]$. The condition that all $k \times k$ minors
of $A$ are nonsingular is equivalent to the condition that the polynomials $P(\cS)$ are linearly independent
over $\F(\A)$ (here, we view the polynomials as elements of $\F(\A)[x]$ instead of as elements of $\F[\A,x]$).
If this is the case, then one can use the Schwartz-Zippel
lemma and show that the formal variables $a_1,\ldots,a_n$ can be replaced with distinct field elements from $\F$, while
still maintaining the property that all $k \times k$ minors of $A$ are nonsingular.
The bound on the field size $|\F| \ge n+k-1$ arises from the degrees of the polynomials obtained in the process.
For details we refer to the original paper \cite{dau2014existence}.

This motivated \cite{dau2014existence} to propose the following algebraic conjecture, which implies \Cref{conj:gmmds}.
\begin{conjecture}[Algebraic GM-MDS conjecture \cite{dau2014existence}]
\label{conj:gmmds-alg}
Let $S_1,\ldots,S_k \subset [n]$ be a set system which satisfies the MDS condition, and where $|S_i|=k-1$ for all $i$.
Then the set of polynomials $P(\cS)$ are linearly independent over $\F(\A)$.
\end{conjecture}

We remark that given any polynomials $p_1,\ldots,p_k \in \F[\A,x]$ (for example, the polynomials appearing in $P(\cS)$),
an equivalent condition to the polynomials
being linearly independent over $\F(\A)$ is the following: for any polynomials $w_1,\ldots,w_k \in \F[\A]$, not all zero,
it holds that
$$
\sum_{i=1}^k w_i(\A) p_i(\A,x) \ne 0.
$$

Following \cite{dau2014existence}, several works~\cite{halbawi2014distributed,heidarzadeh2017algebraic,yildiz2018further}
attempted to resolve the GM-MDS conjecture. They showed that \Cref{conj:gmmds-alg} holds in several special cases, but the
general case remained open. In this work we prove \Cref{conj:gmmds-alg}, which implies \Cref{conj:gmmds}.

\subsection{A generalized conjecture}
We start by considering a more general condition. Let $v \in \N^n$ be a vector, where $\N=\{0,1,2,\ldots\}$ stands
for non-negative integers. The coordinates of $v$ are $v=(v(1),\ldots,v(n))$. We shorthand $|v|=\sum v(i)$.
Given vectors $v_1,\ldots,v_m \in \N^n$ define $\bigwedge v_i \in \N^n$ to be their coordinate-wise minimum:
$$
\bigwedge_{i \in [m]} v_i := (\min(v_1(1),\ldots,v_m(1)),\ldots,\min(v_1(n),\ldots,v_m(n))).
$$
Note that if $v_1,\ldots,v_m \in \{0,1\}^n$ are indicator vectors of sets $S_1,\ldots,S_m \subset [n]$,
then $\bigwedge v_i$ is the indicator vector of $\cap S_i$.

Given a parameter $k > |v|$ define a set of polynomials in $\F[\A,x]$:
$$
P(k,v) := \left\{\prod_{j \in [n]} (x-a_j)^{v(j)} x^e: e=0,\ldots,k-1-|v|\right\}.
$$
Note that $P(k,v)$ consists of $k-|v|$ polynomials of degree $\le k-1$, which form a basis for the linear space of polynomials
of degree $\le k-1$ which have $v(j)$ roots at each $a_j$. Furthermore, note that
if $v$ is the indicator vector of a set $S \subset [n]$
of size $|S|=k-1$, then $P(k,v) = \{p(S)\}$. Given a set of vectors $\cV=\{v_1,\ldots,v_m\} \subset \N^n$ define
$$
P(k,\cV) := P(k,v_1) \cup \ldots \cup P(k,v_m).
$$
We use in this paper the convention that set union can result in a multiset. So for example, if the same polynomial appears in multiple $P(k,v_i)$
then it appears multiple times in $P(k,\cV)$. Under this assumption we always have the identity:
$$
|P(k,\cV)| = |P(k,v_1)| + \ldots + |P(k,v_m)|.
$$
The following definition is the natural extension of the MDS condition to vectors.

\begin{definition}[Property $V(k)$]
Let $\cV=\{v_1,\ldots,v_m\} \subset \N^n$ and $k \ge 1$ be an integer. We say that $\cV$ satisfies $V(k)$ if it satisfies:
\begin{enumerate}
\item[(i)] $|v_i| \le k-1$ for all $i \in [m]$.
\item[(ii)] For all $I \subseteq [m]$ nonempty,
$\sum_{i \in I} (k-|v_i|) + \left|\bigwedge_{i \in I} v_i\right| \le k$.
\end{enumerate}
\end{definition}
Note that when $m=k$ and $v_1,\ldots,v_k$ are indicators of sets $S_1,\ldots,S_k \subset [n]$ of size $|S_i|=k-1$, then property $V(k)$ is equivalent
to the MDS condition for $S_1,\ldots,S_k$.

Observe that in general, if $\cV$ satisfies $V(k)$ then $P(k,\cV)$ contains $\sum_{i=1}^m (k-|v_i|) \le k$ polynomials
of degree $\le k-1$. The following conjecture is the natural extension of \Cref{conj:gmmds-alg} to vectors.

\begin{conjecture}
\label{conj:vecmds}
Let $\cV \subset \N^n$ and $k \ge 1$. Assume that $\cV$ satisfies $V(k)$.
Then the polynomials in $P(k,\cV)$ are linearly independent over $\F(\A)$.
\end{conjecture}
A clarifying remark: as we view the set $P(k,\cV)$ as a multiset, \Cref{conj:vecmds} (and \Cref{thm:vecstarmds} below) imply in particular that the polynomials in $P(k,\cV)$ are all distinct, so $P(k,\cV)$ is in fact a set.

\subsection{An intermediate case}

We prove \Cref{conj:vecmds} under an additional assumption, which is sufficient to prove \Cref{conj:gmmds}. It is still open
to prove \Cref{conj:vecmds} in full generality.

\begin{definition}[Property $V^*(k)$]
Let $\cV=\{v_1,\ldots,v_m\} \subset \N^n$ and $k \ge 1$ be an integer. We say that $\cV$ satisfies $V^*(k)$ if
it satisfies $V(k)$, and additionally it satisfies:
\begin{enumerate}[(i)]
\item[(iii)] $v_i \in \{0,1\}^{n-1} \times \N$ for all $i \in [m]$. Namely, all coordinates in $v_i$, except perhaps the last,
are in $\{0,1\}$.
\end{enumerate}
\end{definition}

\begin{theorem}
\label{thm:vecstarmds}
Let $\cV \subset \N^n$ and $k \ge 1$. Assume that $\cV$ satisfies $V^*(k)$. Then the polynomials $P(k,\cV)$ are linearly independent over $\F(\A)$.
\end{theorem}

\Cref{conj:gmmds-alg} follows directly from \Cref{thm:vecstarmds}. If $S_1,\ldots,S_k \subset [n]$
are sets which satisfy the assumptions of \Cref{conj:gmmds-alg}, then their indicator
vectors $v_1,\ldots,v_k \in \{0,1\}^n$ satisfy the assumptions of \Cref{thm:vecstarmds},
and hence $P(\{S_1,\ldots,S_k\})=P(k,\{v_1,\ldots,v_k\})$ are linearly independent over $\F(\A)$.

\subsection{General distance}
The rows of a $k \times n$ MDS matrix generates a linear code in $\F^n$ whose minimal distance is $d=n-k+1$. Namely, any vector
in the subspace spanned by the rows has at most $n-d=k-1$ zeros. One can ask a more general question: given parameters $k \le n$ and $d \le n-k+1$,
what are the necessary and sufficient conditions on the zero pattern of a code with minimal distance $d$.

As it turns out, this more general question reduces to the one about MDS codes.

\begin{corollary}
Let $k \le n$ and $d \le n-k+1$. Let $S_1,\ldots,S_k \subseteq [n]$. A necessary condition for the existence of a $k \times n$ matrix $A$ over any field,
such that the code spanned by the rows of the matrix has minimal distance at least $d$, and such that $A_{i,j}=0$ whenever $j \in S_i$, is
$$
|I| + \left| \bigcap_{i \in I} S_i \right| \le n-d+1 \qquad \forall I \subseteq [k], I \text{ nonempty}.
$$
It is also a sufficient condition over any field $\F$ of size $|\F| \ge 2n-d$.
\end{corollary}

\begin{proof}
We first show that the conditions are necessary. Assume the condition is violated for some $I$. Then there are $|I|$ rows with at least $n-d+2-|I|$ common zeros. Pick any $|I|-1$ other coordinates; there is
some linear combination of the rows in $I$ which is zero in these coordinates. So this linear combination has $(n-d+2-|I|)+(|I|-1) = n-d+1$ many zeros, a contradiction
to the minimal distance being at least $d$.

To show that the conditions are sufficient, consider the set system $S_1,\ldots,S_k, S_{k+1}=\ldots=S_{n-d+1} = \emptyset$. It satisfies that
$$
|I| + \left| \bigcap_{i \in I} S_i \right| \le n-d+1 \qquad \forall I \subseteq [n-d+1], I \text{ nonempty}.
$$
The claim follows by applying \Cref{thm:main} to this set system.
\end{proof}

\subsection{Related work}
As we already discussed, the GM-MDS conjecture was suggested by
\cite{dau2014existence}, and partial results were obtained by~\cite{halbawi2014distributed,heidarzadeh2017algebraic,yildiz2018further}.
Shortly after posting this result in arXiv~\cite{lovett2018proof}, we were informed by Yildiz and Hassibi~\cite{yildiz2018optimum} 
that they too have found a proof of the GM-MDS conjecture. Inspecting their proof, 
it is similar in spirit to our proof, in the sense that both proofs generalize the original GM-MDS conjecture, in order to facilitate an
inductive argument. More specifically, our approach is to allow multiple roots at a distinguished point, 
while their approach is to allow general multiplicities of sets.

\subsection{Open problems}
We already discussed \Cref{conj:vecmds}. A more general open problem is the following. Let $S_1,\ldots,S_k \subset [n]$ be a set system.
Let $A$ be a $k \times n$ matrix over some field, such that $A_{i,j}=0$ whenever $j \in S_i$. If we make no assumptions on the set system,
then some $k \times k$ minors of $A$ are forced to be singular (this happens when the set system, restricted to the minor, violates the MDS condition). 
The question is: what is the minimal field size, for which there exists a matrix where all minors which are not forced to be singular are nonsingular.

This question arises naturally in the study of Maximally Recoverable (MR) codes, where the minors which are forced to be 
singular are determined by the underlying topology of the code. 
The GM-MDS conjecture which we prove is the special case where no minor is forced to be singular. 
In this case, very small fields (of size $n+k-1$) are sufficient.
However, in general there is no reason for nice algebraic constructions to exist. 
Two recent works~\cite{kane2017independence,gopi2017maximally} have shown that in specific situations, exponential field size is needed. However,
the proof techniques are highly specialized to these specific cases.

This raises the following natural conjecture: most set systems require exponential field size.

\begin{conjecture}
Let $S_1,\ldots,S_k \subset [n]$ be chosen randomly, by including each $j \in S_i$ independently with probability $1/2$. 
Assume that there exists a $k \times n$ matrix $A$ over a field $\F$ that satisfies:
\begin{enumerate}[(i)]
\item $A_{i,j}=0$ whenever $j \in S_i$.
\item Any $k \times k$ minor of $A$, which is not forced to be singular by (i), is nonsingular.
\end{enumerate} 
Then with high probability over the choice of the set system, $|\F| \ge c {n \choose k}^c$, where $c>0$ is some absolute constant.
\end{conjecture} 

The conjecture basically says that for most set systems, the probabilistic construction which requires field size ${n \choose k}$ cannot be
significantly improved.

\paragraph*{Acknowledgement.} I thank Hoang Dau and Sankeerth Rao for a careful reading of an earlier version of this paper.

\section{Proof of \Cref{thm:vecstarmds}}

Let $n,k \ge 1$. Let $\cV=\{v_1,\ldots,v_m\} \subset \N^n$ which satisfies $V^*(k)$. We will prove that the polynomials $P(k,\cV)$ are linearly
independent over $\F(\A)$.

To that end, we assume that $\cV$ is a minimal counter-example and derive a contradiction.
Concretely, the underlying parameters are $n,k,m$ and $d=|P(k,\cV)| = \sum k-|v_i|$. We will assume that if $\cV'$ is a set of vectors
with corresponding parameters $n' \le n,k' \le k,m' \le m,d' \le d$ with at least one of the inequalities being sharp, then \Cref{thm:vecstarmds}
holds for $\cV'$. In particular, we assume that $m \ge 2$, as \Cref{thm:vecstarmds} clearly holds when $m=1$.

To help the reader, we note that the following lemmas construct such $\cV'$ with the following parameters:
\begin{itemize}
\item \Cref{lemma:minb_zero}: $n,k-1,m,d$.
\item \Cref{lemma:tight}: $n,k,e,d'$ and $n,k,m-e+1,d''$ with $2 \le e \le m-1$ and $d',d''<d$.
\item \Cref{lemma:onevec}: $n-1,k,m,d$.
\item \Cref{lemma:n_equals_k}: $n,k,m,d-1$.
\end{itemize}
We use the following notation to simplify the presentation:
$$
v_I := \bigwedge_{i \in I} v_i \qquad I \subseteq [m].
$$
We introduce sometimes in the proofs an auxiliary set $\cV'=\{v'_1,\ldots,v'_{m'}\}$, in which case $v'_I$ for $I \subseteq [m']$ are defined analogously. Below, when we say that $\cV$ or $\cV'$ satisfy (i), (ii) or (iii), we mean the relevant items in the definition of $V^*(k)$.

Given two vectors $u,v \in \N^n$ we denote $u \le v$ if $u(i) \le v(i)$ for all $i \in [n]$.

\begin{lemma}
\label{lemma:dominate}
There do not exist distinct $i,j \in [m]$ such that $v_i \le v_j$.
\end{lemma}

\begin{proof}
Assume the contrary. Applying (i) to $j$ gives $|v_j| \le k-1$. Applying (ii) to $I=\{i,j\}$ gives
$$
(k-|v_i|) + (k-|v_j|) + |v_i \wedge v_j| \le k.
$$
As $v_i \le v_j$ we have $v_i \wedge v_j = v_i$, and hence obtain that $k-|v_j| \le 0$, a contradiction.
\end{proof}

\Cref{lemma:dominate} implies in particular that $n \ge 2$. This is since if $n=1$ then necessarily $m=1$,
as otherwise there would be $i,j$ for which $v_i \le v_j$. So we assume $n \ge 2$ from now on.

\begin{lemma}
\label{lemma:minb_zero}
$\bigwedge_{i \in [m]} v_i = 0$.
\end{lemma}

\begin{proof}
Assume not. Then there exists a coordinate $j \in [n]$ with $v_i(j) \ge 1$ for all $i \in [m]$.
Define a new set of vectors $\cV'=\{v'_1,\ldots,v'_m\} \subset \N^n$ as follows:
$$
v'_i := (v_i(1),\ldots,v_i(j-1),v_i(j)-1,v_i(j+1),\ldots,v_i(n)).
$$
In words, $v'_i$ is defined from $v_i$ by decreasing coordinate $j$ by $1$.

We first show that $\cV'$ satisfies $V^*(k-1)$. Note that $|v'_i|=|v_i|-1$.
It clearly satisfies (i),(iii). To show that it satisfies (ii) let $I \subseteq [m]$.
We have
$$
\sum_{i \in I} (k-1-|v'_i|) + |v'_I| = \sum_{i \in I} (k-|v_i|) + |v_I| - 1 \le k-1.
$$

As we showed that $\cV'$ satisfies $V^*(k-1)$, the minimality of $\cV$ implies that the polynomials $P(k-1,\cV')$ are linearly
independent over $\F(\A)$. The lemma follows as it is simple to verify that
$$
P(k,\cV) = \{p(\A,x) (x-a_j): p \in P(k-1,\cV')\}.
$$
In particular, the linear independence of $P(k-1,\cV')$ implies the linear independence of $P(k,\cV)$.
\end{proof}

\begin{definition}[Tight constraint]
A set $I \subseteq [m]$ is \emph{tight for $\cV$} if property (ii) holds with equality for $I$. Namely if
$$
\sum_{i \in I} (k-|v_i|) + |v_I| = k.
$$
\end{definition}

Note that if $|I|=1$ then $I$ is always a tight constraint.
The following lemma is an extension of Lemma 2(i) in \cite{yildiz2018further}. It shows that in a minimal counter-example
there are no tight sets, except for singletons and perhaps the whole set.

\begin{lemma}
\label{lemma:tight}
If $I \subseteq [m]$ is a tight constraint, then $|I|=1$ or $|I|=m$.
\end{lemma}

\begin{proof}
Assume towards a contradiction that there exist a tight $I$ with $1<|I|<m$.
We will use the minimality of $\cV$ to derive a contradiction. Assume for simplicity
of notation that $I=\{e,\ldots,m\}$ for $2 \le e \le m-1$. Define a new set of vectors
$\cV'=\{v'_1,\ldots,v'_e\}$ given by
$$
v'_1 := v_1, \ldots, v'_{e-1} := v_{e-1}, v'_e := v_I.
$$

We first show that $\cV'$ satisfies $V^*(k)$. It clearly satisfies (i) and (iii). To see that it satisfies (ii) let $I' \subseteq [e]$.
If $e \notin I'$ then $\cV'$ satisfies (ii) for $I'$ as it is same condition as for $\cV$, so assume $e \in I'$.
Let $I'' = I' \cup \{e+1,\ldots,m\}$. Then
$$
\sum_{i \in I'} (k-|v'_i|) + |v'_{I'}| = \sum_{i \in I''} (k-|v_i|) + |v_{I''}| \le k,
$$
where the equality holds since $k-|v'_e| = \sum_{i \in I} (k-|v_i|)$ since we assume $I$ is tight, and since by definition of $I''$ we have $v'_I=v_{I''}$.

As we assume that $\cV$ is a minimal counter-example for \Cref{thm:vecstarmds}, the theorem holds for $\cV'$. So,
the polynomials $P(k,\cV')$ are linearly independent. Observe that $|P(k,\cV')|=|P(k,\cV)|$ since
$$
|P(k,\cV')| = \sum_{i \in [e]} (k-|v'_i|) = \sum_{i \in [m]} (k-|v_i|) = |P(k,\cV)|.
$$
Thus, it will suffice to prove that $P(k,\cV)$ and $P(k,\cV')$ span the same space of polynomials over $\F(\A)$. To that end,
it suffices to prove that $F:=P(k,\{v_e,\ldots,v_m\})$ and $F':=P(k,v'_e)$ span the same space of polynomials.

Let us shorthand $v=v'_e$.
Define the polynomial $p(\A,x):=\prod_{j \in [n]} (x-a_j)^{v(j)}$. Observe that $p$ divides all polynomials in $F,F'$.
Moreover, $F'=\{p(\A,x) x^d: d=0,\ldots,k-1-|v|\}$ spans the linear space of all multiples of $p$ of degree $\le k-1$.
As $|F|=|F'|$ it suffices to prove that $F$ are linearly independent over $\F(\A)$, as then they must span the same linear space.
However, this follows from the minimality of $\cV$, since $F=P(k,\cV'')$ for $\cV''=\{v_e,\ldots,v_m\}$.
\end{proof}

The following lemma identifies a concrete vector that must exist in a minimal counter-example. It is in its proof
that we actually use the assumption that $\cV$ satisfies (iii), namely $V^*(k)$ and not merely $V(k)$.

\begin{lemma}
\label{lemma:onevec}
There exists $i \in [m]$ such that $v_i=(1,\ldots,1,0)$.
\end{lemma}

\begin{proof}
\Cref{lemma:minb_zero} guarantees that there exists $i^* \in [m]$ for which $v_{i^*}(n)=0$.
We will prove that $v_{i^*}=(1,\ldots,1,0)$. If not, then by (iii) there exists $j^* \in [n-1]$ be such that $v_{i^*}(j^*)=0$.
For simplicity of notation assume that $i^*=m,j^*=n-1$.
Define a new set of vectors $\cV'=\{v'_1,\ldots,v'_m\} \subset \N^{n-1}$
as follows:
$$
v'_i := \left(v_i(1),\ldots,v_i(n-2), v_i(n-1)+v_i(n)\right).
$$
In words, $v'_i \in \N^{n-1}$ is obtained by adding the last two coordinates of $v_i \in \N^n$.

We first show that $\cV'$ satisfies $V^*(k)$. Note that $|v'_i|=|v_i|$.
It clearly satisfies (i),(iii). To show that it satisfies (ii) let $I \subseteq [m]$.
Note that (ii) always holds if $|I|=1$, so we may assume $|I|>1$. We have by definition
\begin{equation}
\label{eq:onevec}
\sum_{i \in I} (k-|v'_i|) + |v'_I| - v'_I(n-1) = \sum_{i \in I} (k-|v_i|) + |v_I|
- v_I(n-1) - v_I(n) .
\end{equation}

First, consider first the case where $|I|<m$. \Cref{lemma:tight} gives that $I$ is not tight, and hence
$$
\sum_{i \in I} (k-|v_i|) + |v_I| \le k-1.
$$
As $\cV$ satisfies (iii) we have $v_i(n-1) \in \{0,1\}$ for all $i$. This implies $v_I(n-1) \in \{0,1\}$
and $v'_I(n-1) \in \{v_I(n), v_I(n)+1\}$. So \Cref{eq:onevec} gives
$$
\sum_{i \in I} (k-|v'_i|) + |v'_I| \le \sum_{i \in I} (k-|v_i|) + |v_I|+1 \le k.
$$

Next, consider the case of $|I|=m$. As $v_m(n-1)=v_m(n)=0$ we have $v'_m(n-1)=0$ and hence
$v_I(n-1)=v_I(n)=v'_I(n-1)=0$.
\Cref{eq:onevec} then gives
$$
\sum_{i \in I} (k-|v'_i|) + |v'_I| = \sum_{i \in I} (k-|v_i|) + |v_I| \le k.
$$

As we showed that $\cV'$ satisfies $V^*(k)$, the minimality of $\cV$ implies that the polynomials $P(k,\cV')$ are linearly
independent over $\F(\A)$. We next show that this implies that $P(k,\cV)$ are also linearly independent over $\F(\A)$.

Let $s_i := k-|v_i|$ for $i \in [m]$.
We have $P(k,\cV)=\{p_{i,e}: i \in [m], e \in [s_i]\}$ and $P(k,\cV')=\{p'_{i,e}: i \in [m], e \in [s_i]\}$ where
\begin{align*}
&p_{i,e}(\A,x) := x^{e-1} \prod_{j \in [n-2]} (x-a_j)^{v_i(j)} \; \cdot \; (x-a_{n-1})^{v_i(n-1)} (x-a_n)^{v(n)}\;, \\
&p'_{i,e}(\A,x) := x^{e-1} \prod_{j \in [n-2]} (x-a_j)^{v_i(j)} \; \cdot \; (x-a_{n-1})^{v_i(n-1)+v_i(n)}\;.
\end{align*}
Observe that $p'_{i,e}$ can be obtained from $p_{i,e}$ by substituting $a_{n-1}$ for $a_n$. Namely
$$
p'_{i,e}(a_1,\ldots,a_{n-1},x) = p_{i,e}(a_1,\ldots,a_{n-1},a_{n-1},x).
$$
Assume towards a contradiction that $\{p_{i,e}\}$ are linearly dependent over $\F(\A)$. Equivalently, there exist polynomials $w_{i,e}(\A)$, not all zero,
such that
$$
\sum_{i \in [m]} \sum_{j \in [s_i]} w_{i,e}(\A) p_{i,e}(\A,x) = 0.
$$
We may assume that the polynomials $\{w_{i,e}\}$ do not all have a common factor, as otherwise we can divide them by it. Let $w'_{i,e}(\A)$
be obtained from $w_{i,e}(\A)$ by substituting $a_{n-1}$ for $a_n$. That is, $w'_{i,e}(a_1,\ldots,a_{n-1})=w_{i,e}(a_1,\ldots,a_{n-1},a_{n-1})$.
Then we obtain
$$
\sum_{i \in [m]} \sum_{j \in [s_i]} w'_{i,e}(\A) p'_{i,e}(\A,x) = 0.
$$
As the polynomials $\{p'_{i,e}\}$ are linearly independent over $\F(\A)$, this implies that $w'_{i,e} \equiv 0$ for all $i,e$.
That is, the polynomials $w_{i,e}$ satisfy
$$
w_{i,e}(a_1,\ldots,a_{n-1},a_{n-1}) \equiv 0.
$$
This implies that $(a_{n-1}-a_n)$ divides $w_{i,e}$ for all $i,e$, which is a contradiction to the assumption that $\{w_{i,e}\}$
do not all have a common factor.
\end{proof}

\Cref{lemma:onevec} implies that the vector $(1,\ldots,1,0)$ belongs to $\cV$. Without loss of generality,
we may assume that it is $v_m$. This implies that $v_i(n) \ge 1$ for all $i \in [m-1]$, as otherwise
we would have $v_i \le v_m$, violating \Cref{lemma:dominate}.

\begin{lemma}
\label{lemma:n_equals_k}
$n=k$.
\end{lemma}

\begin{proof}
Let $v_m=(1,\ldots,1,0)$. We know by (i) that $n-1 = |v_m| \le k-1$, so $n \le k$. Assume towards a contradiction that $n<k$.
Define a new set of vectors $\cV'=\{v'_1,\ldots,v'_m\} \subset \N^n$ as follows:
$$
v'_1 := v_1, \ldots, v'_{m-1} := v_{m-1}, v'_m := (1,\ldots,1,1).
$$
In words, we increase the last coordinate of $v_m$ by $1$.

We claim that $\cV'$ satisfies $V^*(k)$. It satisfies (i) by our assumption that $|v'_m|=n \le k-1$, and it satisfies (iii)
clearly. To show that it satisfies (ii), let $I \subseteq [m]$. If $m \notin I$ then
it clearly satisfies (ii) for $I$, as it is the same constraint as for $\cV$, so assume $m \in I$. In this case we have
$$
\sum_{i \in I} (k-|v'_i|) + |v'_I| = \left( \sum_{i \in I} (k-|v_i|) - 1 \right) + \left( |v_I| + 1 \right) \le k.
$$

Note that $|P(k,\cV')|=|P(k,\cV)|-1$. As $\cV$ is a minimal counter-example, we have that $\cV'$ satisfies $V^*(k)$. Let $p(\A,x):=\prod_{j \in [n-1]}(x-a_j)$.
The construction of $\cV'$ satisfies that $P(k,\cV)$ and $P(k,\cV') \cup \{p\}$ span the same linear space of polynomials over $\F(\A)$. This is since
$v'_i=v_i$ for $i=1,\ldots,m-1$ and since
$$
P(k,\{v_m\}) = \{p x^e: e=0,\ldots,n-k\}
$$
and
$$
P(k,\{v'_m\}) \cup \{p\} = \{p (x-a_n) x^e: e=0,\ldots,n-k-1\} \cup \{p\}
$$
both span the linear space of polynomials which are multiples of $p$ and of degree $\le k-1$.

Denote for simplicity of presentation the polynomials of $P(k,\cV')$ by $p_1,\ldots,p_{d-1}$, where $d=|P(k,\cV)|$.
Assume that the polynomials $P(k,\cV)$ are linearly dependent. As $P(k,\cV')$ are linearly independent, it implies that there exist
polynomials $w,w_1,\ldots,w_{d-1} \in \F[\A]$, where $w \ne 0$, such that
$$
w(\A)p(\A,x) + \sum_{i=1}^{d-1} w_i(\A) p_i(\A,x) \equiv 0.
$$
Note that by construction, $v'_i(n) \ge 1$ for all $i \in [m]$. This implies that $p_1,\ldots,p_{d-1}$ are all divisible
by $(x-a_n)$, while $p$ is not. Substituting $x=a_n$ then gives $w \equiv 0$, a contradiction.
\end{proof}

We can now reach a contradiction to $\cV$ being a counter-example. We know that $v_m=(1,\ldots,1,0)$
with $|v_m|=n-1=k-1$.
Let $\cV'=\{v_1,\ldots,v_{m-1}\}$. As it satisfies $V^*(k)$ we have that the polynomials
$P(k,\cV')$ are linearly independent. Moreover, as $|v_m|=k-1$ we have $P(k,v_m) = \{p\}$
where $p(\A,x)=\prod_{j \in [n-1]} (x-a_j)$.
Note that all polynomials in $P(k,\cV')$ are divisible by $(x-a_n)$, while $p$ is not.
So by the same argument as in the proof of \Cref{lemma:n_equals_k}, $P(k,v_m)$ cannot be linearly dependent of $P(k,\cV')$.
So $P(k,\cV)$ are linearly independent.

\bibliographystyle{alpha}
\bibliography{GMMDS}

\end{document}